\documentclass[11pt]{article}

\usepackage[utf8]{inputenc} 
\usepackage[T1]{fontenc}    
\usepackage[margin=1in]{geometry}
\usepackage{setspace}
\usepackage{hyperref}       
\usepackage{url}            
\usepackage{booktabs}       
\usepackage{amsfonts}       
\usepackage{nicefrac}       
\usepackage{microtype}      
\usepackage{natbib}
\usepackage{authblk}

\usepackage{tikz}
\usepackage{algorithm}
\usepackage{algorithmic}
\usepackage{pifont} 
\usepackage{subcaption}

\usepackage{chrkroer}
\usepackage{amssymb}
\usepackage{latexsym}
\usepackage{amsthm}
\newtheorem{claim}{Claim}[]
\newcommand{\rfwa}[0]{RFWA}

\title{Robust and fair work allocation}

\author[1]{Amine Allouah}
\author[2]{Christian Kroer}
\author[1]{Xuan Zhang \thanks{Corresponding Author: Xuan Zhang, xuanzh@fb.com}}
\author[1]{Vashist Avadhanula}
\author[1]{Nona Bohanon}
\author[1]{Anil Dania}
\author[1]{Caner Gocmen}
\author[1]{Sergey Pupyrev}
\author[1]{Parikshit Shah}
\author[1]{Nicolas Stier}
\author[1]{\\Ken Rodr\'{i}guez Taarup}
\affil[1]{Meta}
\affil[2]{IEOR Department, Columbia University}
\date{}   
\begin{document}

\maketitle

\begin{abstract}
    In today's digital world, interaction with online platforms is ubiquitous, and thus content moderation is important for protecting users from content that do not comply with pre-established community guidelines. Given the vast volume of content generated online daily, having a robust content moderation system throughout every stage of planning is particularly important. We study the short-term planning problem of allocating human content reviewers to different harmful content categories. We use tools from fair division and study the application of competitive equilibrium and leximin allocation rules for addressing this problem. On top of the traditional Fisher market setup, we additionally incorporate novel aspects that are of practical importance. The first aspect is the forecasted workload of different content categories, which puts constraints on the allocation chosen by the planner. We show how a formulation that is inspired by the celebrated Eisenberg-Gale program allows us to find an allocation that not only satisfies the forecasted workload, but also fairly allocates the remaining working hours from the content reviewers among all content categories. The resulting allocation is also robust in a sense that the additional allocation provides a guardrail in cases where the actual workload deviates from the predicted workload. The second practical consideration is time dependent allocation that is motivated by the fact that partners need scheduling guidance for the reviewers across days to achieve efficiency.
    To address the time component, we introduce new extensions of the various fair allocation approaches for the single-time period setting, and we show that many properties extend in essence, albeit with some modifications. Lastly, related to the time component, we additionally investigate how to satisfy markets' desire for smooth allocation -- that is, partners for content reviewers prefer an allocation that does not vary much from time to time, so that the switch in staffing is minimized. We demonstrate the performance of our proposed approaches through real-world data obtained from Meta.
\end{abstract}

\section{Introduction}

Content moderation is an important challenge faced by many online platforms, and it typically involves the analysis of content generated on the platform to verify whether it is compliant with the platform's content policies.
Examples of problematic content categories include fake profiles, spam, hate speech, violent content, and harassment~\citep{facebook-cs}.
A key issue in this domain is the enormous amount of content that needs to be verified.
For example, Meta reported in Q3 of 2021: ``In Q3, the prevalence of hate speech on Facebook was 0.03\% or 3 views of hate speech per 10,000 views of content'' \citep{rosen2021}; and across Facebook and Instagram, actions were taken on over 100 million pieces of harmful content across all categories \citep{meta-csep}.
While artificial intelligence technology is necessarily used to handle this volume of content, it is not, at least currently, accurate enough to handle content moderation without human help.
For this reason, Meta employs thousands of people to help analyze ambiguous content that requires human verification or investigation~\citep{osofsky2019}.
Across the internet industry, social media companies, such as Twitter, YouTube, and Facebook, spend billions of dollars every year contracting with different partners that provide content moderating services~\citep{cnbc-cm}.
In this paper, we study the short-term planning problem of allocating reviewing capacity to different harmful content categories, while taking into account uncertainty about content volume as well as the skills (e.g., language requirement, training on community standards) of various partners supplying reviewer capacity.

We introduce the \emph{robust and fair work allocation} (RFWA) problem. 
In the RFWA problem we have a set of $n$ content categories, each with some predicted workload $d_i$, referred to as the \emph{demand} of category $i$.
We also have a set of $m$ partners for content reviewers, where each partner $j$ has some capacity $s_j$ for reviewing content; we use $s_j$ because we will see that this capacity is analogous to supply in the context of markets.
For every partner $j$, consider all and \emph{only} the content categories for which the partner has the appropriate language and cultural context as well as necessary training to work for. For each of these content category $i$, there is a value $v_{ij}$
which we will mostly think of as the rate of partner $j$ to handle category $i$ 
but it can also be used to model, for example, the accuracy of the work.
The goal of the planner is to produce an allocation of workload capacity of partners to categories, in a way that satisfies both every predicted demand and every capacity constraint.
In general, we will be interested in the setting where there is more work capacity than demand, in which case we would like to distribute excess work capacity across the content categories in a way that treats each category ``fairly,'' since each category has variable demand realizations, and have different stakeholders that care particularly about a given category.

To address the \rfwa\ problem, we propose the use of tools from fair division.
In particular, we study the application of leximin, competitive equilibrium from equal incomes (CEEI)~\citep{varian1974equity}, and more generally the max Nash welfare (MNW) allocation rule~\citep{caragiannis2019unreasonable}, for approaching this problem.
However, as we will discuss next, the \rfwa\ has many practical considerations that necessitate the development of extensions to these tools.

The first practical consideration of the \rfwa\ problem is the demands of the categories.
In our model, these demands are hard constraints that must be satisfied, and the ``utility'' achieved by each category is the supply of excess reviewing capacity that it receives beyond the forecasted demand. We refer to this excess allocation as \emph{over-allocation}, and the primary goal of the planner is to fairly over-allocate.
This requires us to consider an allocation model that is similar to the Nash bargaining model~\citep{nash1953two}, where each agent has a \emph{disagreement} point, which is analogous to the demands of categories in our model.
We introduce variations on market-equilibrium-based allocation (including CEEI) and leximin fairness for handling these demands. We then show that, analogous to results for the standard setting~\citep{halpern2020fair,aziz2020almost}, we have equivalence between leximin and CEEI when all categories have \emph{binary valuations}, i.e. partners can either handle a category or not; a second useful market-equilibrium approach where the budgets are not uniform leads to a different type of allocation, however.

A second practical consideration is that of \emph{time}. 
As discussed in prior work on other aspects of the content moderation problem, there is an important time component to this problem~\citep{nguyen2020clara,makhijani2021quest}.
The reason for this is that the time between the creation of a piece of content and its review can have a large impact on the platform quality. If harmful content is visible to users before it has been reviewed, then it negatively impacts users.
Conversely, if a piece of policy-compliant content is left in the review queue for a long time, it negatively impacts both the content creator and users.
To address this time component, we introduce new versions of the various fair allocation approaches that take into account this time component. 

Finally, related to the time component is the issue of smooth allocation across time.
In practice, the different partners prefer to have relatively smooth allocations from each category across time. This makes it easier for them to perform staff allocation, as well as staff training for dealing with each of the categories.
We introduce a new variation of the market equilibrium model that takes into account this smoothness component.

For each of the variations on the standard models introduced above, we study the properties resulting from these new variations, and show that in most cases we can still give desirable guarantees, for example by showing that market-equilibrium interpretations of the MNW allocation are preserved for several of these extensions.
We also perform extensive numerical tests on real-world content moderation data, and study the qualitative and quantitative results derived from applying each of the proposed models to the \rfwa\ problem.

While our paper is motivated by the \rfwa\ problem, the ideas apply more broadly to several related operational problems.
Here we mention two applications that are quite similar to the content moderation problem. The first one is the staffing problem for customer support services. The buyers here are different types of issues that need to be resolved, and each one has its own (expected) workload. On the other hand, the sellers are customer service representatives, and each representative can only work on a subset of the issues due to training, and has an upper limit on how much they can work in a day or in a week. The second application is how to schedule workers to provide labels for content that will later be used for the training of machine learning algorithms. Again, the buyers here are different types of content, and sellers are the workers. It is not hard to see that for both applications, over-allocation and smooth allocation arise as issues that must be addressed.

\noindent\textbf{Relationship to Market Equilibrium}
So far we have discussed the \rfwa\ problem in its concrete context. However, for the remainder of the paper, we will discuss our results in a more general market equilibrium context, in order to be more consistent with existing literature.
To that end, we now describe how the \rfwa\ problem maps to a market equilibrium problem.
The set of content categories corresponds to the set of buyers in a market (who now possess demands).
The set of partners corresponds to the set of items in the market, and capacities correspond to item supplies.
The value $v_{ij}$ that category $i$ has for partner $j$ maps directly to a valuation in the market.

\begin{itemize}
    \item Over allocation: In some cases, there is more capacity than the expected workload. In such cases, we want to allocate fairly the oversupply of qualified capacities to different eligible content categories.
    \item Smoothness: We want to allocate over multiple time periods, and we want an allocation that is humanly feasible, meaning we can not have an allocation that alternates a lot between consecutive periods (i.e., an allocation that is smooth over time).
\end{itemize}

\subsection{Related Literature}

Content moderation has gained its importance in the last decade as we spend more time browsing on the internet. Hence, many problems that arise from content moderation have been studied both inside and outside academia. The first is on developing reliable machine learning (ML) algorithms to automate content moderation as much as possible (see, e.g., \cite{jhaver2019human,gillespie2020content}). Given that ML alone cannot solve the content moderation problem \citep{cm-ai}, it is also important to efficiently use the human review capacity. \citet{haimovich2020scalable} developed a framework to predict the popularity of social network content in real-time, which could be used to detect harmful viral content and ultimately enable timely content moderation. \citet{nguyen2020clara} developed a statistical framework that increases the accuracy for detecting harmful content, by combining the decisions from multiple reviewers as well as those from ML algorithms. \citet{makhijani2021quest} built simulation models for the large-scale review systems to understand and optimize the human review process and guide the platforms for operational decisions (e.g., hiring additional reviewers). \citet{garcelon2021top} investigated a multi-arm bandit framework to calibrate the severity predictions of content based on various ML models.

To add to the plethora of works on different aspects of content moderation, in this work, we focus on the problem of scheduling reviewers in advance given the anticipated amount of workload, which we model as an allocation problem. The problem of allocating \emph{divisible} goods to buyers fairly was first solved by \citet{varian1974equity} via the so-called \emph{competitive equilibrium from equal incomes} (CEEI) solution. The CEEI solution is fair in a sense that it is envy-free \citep{foley1966resource}. Moreover, the CEEI solutions coincide with allocations that maximize the \emph{Nash social welfare} \citep{arrow2005handbook} and thus can be computed in polynomial time via the convex program given by \citet{eisenberg1959consensus}.

In terms of utility optimization, besides maximizing Nash welfare (MNW), another line of work investigates an egalitarian notion of utilities -- that is, to find allocations which are leximin-optimal. In the case of binary additive valuation, the set of MNW allocations and leximin-optimal allocations are known to coincide \citep{halpern2020fair,aziz2020almost}. Moreover, leximin-optimal solutions have been shown to satisfy several fairness properties for the allocation of indivisible goods (see, e.g., \cite{plaut2020almost,kurokawa2018leximin,freeman2019equitable}).

The over-allocation setting of our paper is closely related to the framework of Nash bargaining (see, e.g., \cite{panageas2021combinatorial}) with the hard demands being interpreted as the \emph{disagreement point}. A key difference between our model and the Nash bargaining model is that in the latter, there is the additional constraint that every buyer is allocated with exactly one unit of items. Interestingly, such fractional matching type of constraints allow the equivalence relation between the MNW allocations and leximin-optimal allocations to be extended to a more general setting of \emph{bi-valued} valuations \citep{aziz2020random}. However, we show in Section \ref{sec:demand-compare} that this result does not hold in our setting.


\section{Fisher Markets with Over allocation and Discrete Time}

We start by introducing the standard \emph{Fisher Market} setting. There is a set of $n$ \emph{buyers} (indexed by $i$), and $m$ \emph{items} (indexed by $j$). Each buyer has some \emph{budget} denoted by $B_i\in \R_{\ge 0}$, and each item has a certain \emph{supply} denoted by $s_j\in \R_{\ge 0}$. Moreover, every buyer $i$ has some utility function $u_i(x_i)$ denoting the utility that they derive from \emph{bundle} $x_i = (x_{i1}, x_{i2}, \cdots, x_{im}) \in \mathbb R_{\ge 0}^m$. We will mostly be interested in \emph{linear utilities}, where $$u_i(x_i) = \langle v_i, x_i \rangle \coloneqq \textstyle \sum_j v_{ij}x_{ij},$$ and $v_i = (v_{i1}, v_{i2}, \cdots, v_{im})\in \mathbb R_{\ge 0}^m$ is buyer $i$'s \emph{valuation}\footnote{If buyer $i$ is not eligible for item $j$ (i.e., in the application of content moderation, if a partner does not have the appropriate language and cultural context or necessary training for a content category), the valuation $v_{ij}$ will be zero.} of the items, although we will also consider generalizations of linear utilities a few times in the paper. We use $x \in \mathbb R_{\ge 0}^{n\times m}$ to denote an \emph{allocation} of items to buyers, with row $x_i$ denoting the bundle allocated to buyer $i$, and $x_{ij}$ denoting how much buyer $i$ gets of item $j$. The \emph{utility profile} of an allocation $x$ is the \emph{set} of utilities of all buyers. Note that in the for utility profile, the ordering of buyers does not matter. 
 
Given a set of \emph{prices} $p\in \mathbb R_{\ge 0}^m$, the \emph{demand set} of a buyer $i$ is the set of optimal bundles within budget: $$D_i(p) = \arg\max_{x_i \geq 0} \left\{ u_i(x_i) | \langle p,x_i \rangle \leq B_i \right\}.$$

A \emph{market equilibrium} is a price-allocation pair $(p,x)$ such that:
\begin{itemize}
    \item Every buyer $i$ receives a bundle from their demand set, i.e. $x_i \in D_i(p)$; and
    \item Any item $j$ with positive price $p_j>0$ must be exactly allocated, i.e. $\sum_i x_{ij} = s_j$.
\end{itemize}

In the standard setting, a market equilibrium is not only guaranteed to exist, it is also efficiently computable via convex programming. The celebrated \emph{Eisenberg-Gale} convex program \citep{eisenberg1959consensus} achieves this:
\begin{equation} \label{prob:eg}
\begin{array}{rrll|l}
    \displaystyle \max_{x\geq 0}\quad & \multicolumn{3}{l}{ \displaystyle\sum_i B_i\log u_i(x_i)} & \textrm{Dual variables}\\
    \textup{s.t.} \quad &\displaystyle  \sum_i x_{ij} &\leq  1, & \forall j=1,\ldots,m\quad & p_j
\end{array}
\tag{EG}
\end{equation}

The market equilibrium arises by combining the optimal primal solution with its corresponding dual variables on the supply constraints. The former gives the allocation that maximizes the \emph{budget-weighted geometric mean of buyers' utilities} (i.e., $[\prod_i u_i^{B_i}]^{\frac{1}{n}}$), and the dual variables act as prices. 

When every buyer has a budget of one (i.e., $B_i=1$), the objective of \eqref{prob:eg} is equivalent to maximizing the \emph{Nash Social Welfare (NSW)}, which is defined as the geometric mean of buyers' utilities. Moreover, market equilibrium is used to define one of the premier methods for fairly allocating divisible goods: \emph{competitive equilibrium from equal incomes} (CEEI).
The reason is that allocation via CEEI leads to several attractive fairness properties:
\begin{itemize}
    \item \emph{Pareto optimality}: there is no other allocation where every buyer is weakly better off (in terms of utilities) and some buyer is strictly better off.
    \item \emph{Envyfreeness}: every buyer $i$ prefers their own bundle to that of any other buyer $k$. If budgets are not equal, then this holds after scaling by budgets: $u_i(x_i) \geq \frac{B_i}{B_k}u_i(x_k)$.
    \item \emph{Proportionality}: every buyer $i$ prefers their own bundle to the one where they receive $\frac{B_i}{\sum_k B_k}$ of every item.
\end{itemize}

In this paper, we are interested in two variations on the standard Fisher market setting.

\paragraph{Allocation with hard demands from buyers.} First, each buyer $i$ has some \emph{required} (or \emph{hard}) utility demand $d_i$. We are required to find an allocation where every buyer $i$ gets utility at least $d_i$, and subject to feasibility, the goal is to \emph{fairly} allocate the excess utility, where the utility by which we measure fairness is $u_i(x_i) = \langle v_i, x_i \rangle - d_i$. We discuss and compare different fairness solution concepts in Section \ref{sec:demand}.

\paragraph{Allocation over multiple time periods.} Secondly, we are interested in settings where there is a temporal component: there are $T$ time steps, and items have supply constraints both per time step, and across all time steps. In particular, let $x_{ij}^t$ denote how much of item $j$ buyer $i$ receives at time $t$. Then, for every time step $t$, we must satisfy the \emph{per-time-step supply constraint} $\sum_i x_{ij}^t \leq s_j^t$, and additionally satisfy the \emph{overall supply constraint} $\sum_t \sum_i x_{ij}^t \leq s_j$. Note that that the temporal component only matters if $s_j < \sum_t s_j^t$, since otherwise we could split each item up into $T$ separate items each with supply specified by its per-time-step supply. We discuss different ways to model buyers' overall utility in Section \ref{sec:multi-time} and show how to satisfy markets' desire for smooth allocation overtime in Section \ref{sec:multi-time-smooth}.

\section{Fair allocation of oversupply} \label{sec:demand}
Here we assume that there's a \emph{hard} demand $d_i$ on utilities for each buyer $i$, and then they have additional \emph{soft} demand for getting additional supply. The goal is to find an allocation that satisfies every hard demand, and furthermore fairly allocate the oversupply to buyers as soft demands. To model fair allocation of oversupply, we define the utility of buyer $i$ given bundle $x_i$ as $$u_i(x_i)\coloneqq \sum_{j} v_{ij} x_{ij}-d_i.$$ 
In this section, we consider two solution concepts: market equilibrium and leximin rules. 

\subsection{Solution via Market Equilibrium}

For market equilibrium, we can apply the same idea as in \eqref{prob:eg} and give the formulation below, which maximizes the budget-weighted geometric mean of buyers' utilities. By scaling, we additionally assume without loss of generality that the supply $s_j=1$ for all sellers. 


\begin{equation}
\begin{array}{rrl|c}\label{prob:main}
  \displaystyle \max_{x\geq 0}\quad & \multicolumn{2}{l}{ \displaystyle \sum_i  B_i\log (\sum_{j} v_{ij} x_{ij} - d_i) } & \text{dual variables} \\[2mm]
  s.t.\quad &\displaystyle  \sum_i x_{ij}  \le  1, & \forall j=1,\ldots,m  & p_j
\end{array}
\tag{EG-demand}
\end{equation}

The objective in the formulation resembles that for Nash bargaining (see, e.g., \cite{panageas2021combinatorial}) with the hard demands being interpreted as the \emph{disagreement point}. However, our problem is fundamentally different from Nash bargaining, which enforces a fractional matching type of constraints and requires that every buyer is allocated to exactly one unit of goods. As a result, the solution spaces between our formulation and the Nash bargaining formulation are different, and results from the Nash bargaining setting do not immediately extend to our setting.

For the rest of the paper, we impose the following regularity assumption, which ensures that the maximum budget-weighted geometric mean of buyers' utilities is strictly positive, and the \eqref{prob:main} problem does not take logarithm of zeros. 

\begin{assumption} \label{ass:pos}
    There exists an allocation $x$ such that $u_i(x_i)>0$ for all buyer $i$. 
\end{assumption}

To simplify the notation for the rest of the section, we let $u_i\coloneqq \sum_{j} v_{ij} x_{ij}-d_i$ be the utility of agent $i$, and let $\lambda_{ij}$ be the dual variable corresponding to the non-negative constraints of \eqref{prob:main}. Then, the Lagrangian of problem \eqref{prob:main} is 
\begin{equation*}
   \sum_i  B_i\log \left( \sum_j v_{ij} x_{ij} - d_i \right) -  \sum_j p_j \left(\sum_i x_{ij} - 1\right) + \sum_{i,j} \lambda_{ij} x_{ij}.
\end{equation*}
By the stationarirty of the KKT conditions, we get that 
\begin{equation} \label{eq:demand-stationary-KKT}
    B_i \frac{  v_{ij}}{u_i} - p_j + \lambda_{ij} = 0 \quad \forall i=1,\dots,n; \, j=1,\dots,m.
\end{equation}

One property of the traditional fisher market (i.e., demand $d_i=0$ for all buyers) is that with price vector $p$, every buyer spends their entire budget. Here, we show a similar property, but with \emph{demand-inflated} budgets.

\begin{lemma} \label{lem:inflated-budget-demand}
    The market prices $(p_j)$  verify
    $$ B_i \left(1+ \frac{d_i}{u_i}\right) = \sum_j p_j x_{ij}$$
\end{lemma}

\begin{proof}
    We omit variables $\lambda_{ij}$'s  and rewrite the stationary condition \eqref{eq:demand-stationary-KKT} as:
    $$B_i \frac{  v_{ij}}{u_i} -   p_j  \ge 0 \quad \forall i=1,\dots,n; \, j=1,\dots,m,$$
    where equality holds whenever $x_{ij}>0$ due to complementary slackness. Multiplying both sides by $x_{ij}$ and summing over $j$, we get that
    $$B_i \frac{  \sum_j v_{ij} x_{ij}}{u_i} =   \sum_j p_j x_{ij}.$$
    Since $\sum_j v_{ij} x_{ij} = u_i + d_i$, we conclude the result. 
\end{proof}

Note that for when buyers do not have hard demands (i.e., $d_i=0$), the equality in Lemma~\ref{lem:inflated-budget-demand} recovers the property that every buyer spends all their budget: $B_i=\sum_j p_j x_{ij}$.

We next show the primal and dual solution pair $(x,p)$ to \eqref{prob:main} forms a market equilibrium under a slightly different definition of demand sets as a consequence of Lemma~\ref{lem:inflated-budget-demand}.

\begin{theorem} \label{thm:mkt-eq-demand}
    Let $x,p$ be a solution to \eqref{prob:main}. The pair $(x,p)$ is a market equilibrium where the demand sets are defined with respect to the demand-inflated budgets. That is, for each buyer $i$, 
    $$D_i(p) \coloneqq \arg \max_{x_i\ge 0} \left\{ \sum_{j} u_i(x_i) \coloneqq v_{ij} x_{ij} - d_i \, \bigg|\, \sum_{j} p_j x_{ij} \le B_i (1+\frac{d_i}{u_i(x_i)}) \right\}.$$
\end{theorem}

\begin{proof}
    The second condition in the definition of market equilibrium is satisfied due to complementary slackness. To show the first condition, which is every buyer $i$ receives a bundle from their demand set, we show that every item in the bundle $x_i$ maximizes the utility per dollar for buyer $i$. By dual feasibility of problem~\ref{prob:main}, we have $\lambda_{ij}\ge 0$. Hence, rearranging the terms of \eqref{eq:demand-stationary-KKT}, we have
    \begin{equation*}
        \frac{v_{ij}}{p_j} \le \frac{u_i}{B_i} \quad \forall i=1,\dots,n; \, j=1,\dots,m,
    \end{equation*}
    where equality holds when $x_{ij}>0$ due to the complementary slackness condition $x_{ij} \lambda_{ij} = 0$. This implies that the utility per dollar is maximized by the buyer's bundle, concluding the proof. 
\end{proof}

Before we continue to the next solution concept, we would like to point out that there are two interesting choices for buyers' budgets and both result in solutions with interesting properties. The first one is when all buyers have budgets of one, and in this case, the objective of \eqref{prob:main} is simply the Nash welfare. The second choice is when buyers' budgets are equal to their demands.

\begin{proposition} \label{prop:fair-add-or-multi}
    Assume $v_{ij}=1$ for all buyers $i$ and sellers $j$. Then,
    \begin{itemize}
        \item[(i)] when $B_i=1$ for all buyers $i$, the difference $\sum_j x_{ij}-d_i$ is constant for all buyers $i$;
        \item[(ii)] when $B_i=d_i$ for all buyers $i$, the ratio $\frac{\sum_j x_{ij}}{d_i}$ is constant for all buyers $i$.
    \end{itemize}
\end{proposition}

\begin{proof}
    First note that the fact every buyer is assigned with a bundle in their demand set, as we have shown in Theorem~\ref{thm:mkt-eq-demand}, implies that all items have the same price, which we denote by $\bar p$. Let $\bar x_i\coloneqq \sum_{j} x_{ij}$ denote the amount of items allocated to buyer $i$. Then the property in Lemma~\ref{lem:inflated-budget-demand} becomes $B_i \frac{\bar x_i}{\bar x_i-d_i} = \bar p \bar x_i$ ($\ddagger$). For (i), we can simplify ($\ddagger$) as $\bar x_i-d_i = \frac{1}{\bar p}$ for all buyers $i$; and for (ii), we can simplify ($\ddagger$) as $\frac{\bar x_i}{d_i} = 1 +  \frac{1}{\bar p}$ for all buyers $i$. These two simplifications conclude the proof.
\end{proof}

For more general non-complete graphs, we anticipate similar results hold in a sense that the solution will try to be as fair as possible, additively when budgets are one and multiplicatively when budgets are equal to demands. Indeed, we observe this phenomenon in our experiments (see Section \ref{sec:app-single-time} for details).

\subsection{Solution via Leximin Fairness Concept}

In this section, we introduce another fairness solution concept, called the \emph{leximin-optimal} allocation.

Consider two vectors of equal length $v=(v_1, v_2, \cdots, v_n)$ and $\tilde v=(\tilde v_1, \tilde v_2, \cdots, \tilde v_n)$. Let $v^*$ denoted an ordered vector of $v$ so that $v^*_1 \le v^*_2 \le \cdots v^*_n$ (ties breaking randomly), and similarly let $\tilde v^*$ be an ordered vector of $\tilde v$. We say $v$ is \emph{leximin-greater} than $\tilde v$, denoted by $v \succ_{\textup{lex}} \tilde v$, if there exists $k\in [n]$ such that $v^*_i = \tilde v^*_i$ for all $i<k$, and $v^*_k > \tilde v^*_k$. 

For any allocation $x$, let $r(x)$ be the vector of budget-weighted utilities, where $r_i(x) = u_i(x_i)^{B_i}$. An allocation $x$ is said to be \emph{leximin-optimal} if there is no other allocation $\tilde x$ such that $r(\tilde x) \succ_{\textup{lex}} r(x)$. In other words, the leximin-optimal solution first maximizes the minimum component in the $r(x)$ vector, then among all allocations where the minimum component in the $r(x)$ vector is maximized, it maximizes the second minimum value, etc.

Leximin is a refinement of the \emph{max-min fairness} and a leximin-optimal allocation can be found by iteratively solving a sequence of max-min convex problems \citep{wilson1998fair}. See Algorithm~\ref{alg:leximin} for details.

\begin{algorithm}[ht] 
\caption{Leximin-optimal solution} \label{alg:leximin}
\begin{algorithmic}[1] 
    \STATE $F \gets [n]$ \COMMENT{set of free variables}
    \STATE Initialize the optimization problem 
    \begin{equation} \label{optProb:leximin}
    \begin{array}{rrl}
      \displaystyle \max_{x\geq 0}\quad & \multicolumn{2}{l}{ \displaystyle \min_{i\in F} r_i(x) }  \\[2mm]
      s.t.\quad & \sum_i x_{ij}  \le  1, & \forall j=1,\ldots,m 
    \end{array}
    \tag{leximin}
    \end{equation}
    \WHILE{$F\neq \emptyset$}
    \STATE Solve \eqref{optProb:leximin} and let $r^*$ be the optimal solution and let $i^*$ be a minimizer.
    \STATE Add to \eqref{optProb:leximin} constraints: $r_{i}(x)\ge r^*$ for all $i\in F$.
    \STATE Remove $i^*$ from $F$
    \ENDWHILE
    \RETURN the latest solution of \eqref{optProb:leximin} 
\end{algorithmic}
\end{algorithm}

When buyers have equal budget and \emph{binary additive utilities}, the leximin-optimal solution is known to maximize NSW, both when the items are indivisible \citep{aziz2020almost} and when the items are divisible \citep{halpern2020fair}. Therefore, we are interested in the connection between the leximin-optimal solutions and the solutions to \eqref{prob:main} (referred to as EG solutions thereafter) in our setting.

\subsection{Comparison} \label{sec:demand-compare}

We say two allocations are \emph{equivalent} if they give the same utility profile. In this section, we investigate whether the EG solutions and the leximin-optimal solution are equivalent under various scenarios (see Table~\ref{tab:summary-demand} for the summary). 

In terms of valuation, we first consider the special case of \emph{binary} valuations (i.e., $v_{ij}\in \{0,1\}, \forall i,j$), which  represents compatibility between buyers and sellers and is commonly used in practice. Before considering the case with general valuations, we consider a slightly more restricted generalization known as bi-valued valuations (i.e., $v_{ij}\in \{\alpha,\beta\}, \forall i,j$, for some $\alpha> \beta\ge 0$). The reason we consider the bi-valued setting is that in other domains such as Nash bargaining, many interesting results arise when valuations are bi-valued (see, e.g., \cite{aziz2020random}). However, as we show in the following, in our setting, the leximin-optimal solution and the EG solution do not coincide for bi-valued valuations. Hence, we no longer need to investigate the most general setting as these negative results directly apply.

Besides different cases of valuations, we also consider two cases for budgets. The first one is when all buyers have equal budgets, for which we assume wlog that every buyer has a budget of one; and the second is when buyers' budgets are equal to their demands. Note that for the first case, EG solutions coincide with the maximum Nash welfare (MNW) solutions, and the proof is similar to the one presented in \cite{halpern2020fair} for binary valuations in the classical setting without hard demand constraints.

\begin{table}[ht]
    \centering
    \begin{tabular}{c|l|l}
         & budget $=1$ & budget $=$ demand \\[1mm]
        \hline
        binary valuation & \ding{51} (Theorem \ref{thm:leximin-eg-equiv}) & \ding{55} (Example~\ref{eg:not-equiv-b-eq-d}) \\
        \hline 
        bi-valued valuation & \ding{55} (Example~\ref{eg:not-equiv-bi-val}) & \ding{55}
    \end{tabular}
    \caption{This table summarizes whether the EG solutions are equivalent to the leximin-optimal solutions under different cases. }
    \label{tab:summary-demand}
\end{table}

\begin{example} \label{eg:not-equiv-bi-val}
This example shows that the EG solution is not the same as the leximin-optimal solution even for bi-valued valuations. Assume we have two items, and two buyers with zero demand and unit budgets. All pairs have valuation one except for one pair: $v_{1,1}=2$ and $v_{1,2}=v_{2,1} =v_{2,2}=1$. Both the EG solution and the leximin-optimal solution allocate the entire second item to the second buyer. However, for the first item, the EG solution allocates the whole item to the first buyer, whereas the leximin-optimal solution allocates only $\frac{2}{3}$ units of it to the first buyer.
\end{example}

\begin{example} \label{eg:not-equiv-b-eq-d}
The following example shows that the EG solution is not the same as the leximin-optimal solution when budget equals to demand, even when the valuation is binary. Consider the instance with two buyers and two items. Each item has a supply of one. The buyers have demand $0.1$ and $0.2$ respectively. The first buyer only accepts the first item, while the second buyer accepts both items.

\begin{figure}[h!]
    \centering
    \begin{tikzpicture}
        \node[rectangle, draw] (1) at (0,1) {buyer 1};
        \node[rectangle, draw] (2) at (3,1) {item 1};
        \node[rectangle, draw] (3) at (0,0) {buyer 2};
        \node[rectangle, draw] (4) at (3,0) {item 2};
        \draw[thick] (1) -- (2) node [midway, fill=white] {1};
        \draw[thick] (3) -- (2) node [midway, fill=white] {1};
        \draw[thick] (3) -- (4) node [midway, fill=white] {1};
    \end{tikzpicture}
\end{figure}

Under the EG solution, the amount of the first item allocated to the first buyer is the solution for
$$\max_{a\in [0,1]} (a-0.1) (2-a-0.2)^2.$$
Therefore, the EG solution allocates $\frac{2}{3}$ units of the first item to the first buyer. However, under the leximin-optimal solution, the amount of the first item allocated to the first buyer is the solution for the equation 
$$a-0.1 = (2-a-0.2)^2,$$
with $a\in [0,1]$ due to the AM-GM inequality. The solution here is clearly not $\frac{2}{3}$.
\end{example}

\begin{theorem} \label{thm:leximin-eg-equiv}
    With binary valuations and unit budget, for \eqref{prob:main}, the set of leximin-optimal solutions is the same as the set of MNW solutions under Assumption~\ref{ass:pos}.
\end{theorem}

We defer the proof to Section \ref{sec:multi-time-sum}, where we introduce a more general setting where the same result holds. That is, Theorem \ref{thm:leximin-eg-equiv} may be viewed as a corollary to the more general result presented in Theorem \ref{thm:leximin-eg-equiv-time-sum}. 

\subsection{Application} \label{sec:app-single-time}

The content quality and integrity of online content platforms play a central role for users experience and platform growth. For that, companies spend billions of dollars every year to contract with different partners that provide human content reviews \citep{cnbc-cm}. In this and the following application sections, we demonstrate through real data acquired from Meta (formerly known as Facebook) how our proposed methods can be useful in allocating partners' capacity to different content categories.

In this setting, the sellers are the \emph{partners} that provide human review services and their supplies represent the amount of reviewer hours they can schedule. The buyers are the different categories of content that need to be reviewed (hereafter referred to as \emph{work types}), and their hard demands are the expected amount of human reviewing hours needed to review all the jobs (i.e., the \emph{forecasted workloads}). 

Since reviewers need to go through extensive training before being able to work on jobs in each work type, not all partners can work on all work types. As a result, we adopt the binary valuation structure to model compatibility between sellers and buyers. That is, if a partner can work on a certain work type, the pair has a valuation of one, and otherwise, it will have a value of zero.

\begin{figure}[ht]
    \centering
    \includegraphics[width=.6\textwidth]{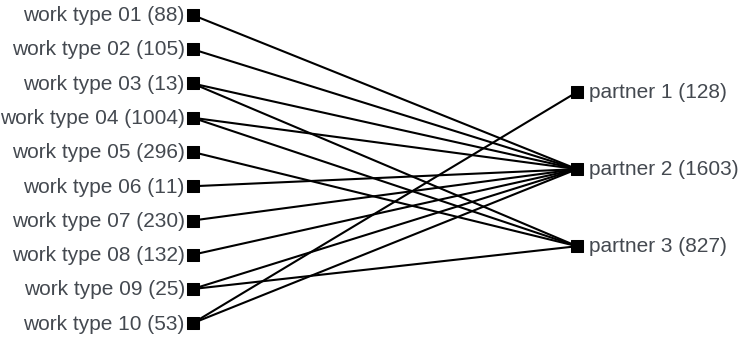}
    \caption{The graph displays buyers (reviewing work types) and sellers (reviewer partners) in a major market, where edges indicate whether a certain partner is able to work on a certain work type. The numbers in parentheses represent demands and supplies, and these numbers are generated through a certain transformation of the raw data. }
    \label{fig:market-arabic}
\end{figure}

For demonstration purposes, we will focus on one particular large-size market for the rest of the paper. See Figure~\ref{fig:market-arabic} for an example of buyers (and their demands), sellers (and their supplies), and the valuation structure for this market.

In the following, we delve into some specific reasons why a fair and robust allocation is important for this business problem. First, since the demands are forecasted workload, the actually realized workload will unavoidably be different from the demand. Thus, being able to schedule the partners to handle more jobs than the forecasted workload provides a guardrail against natural fluctuation in the amount of workload as well as fluctuations due to headline events. 

In addition, as the content review workforce is constantly evolving, occasionally, new work types are created and would have zero demands in the system. However, such zero demand is actually due to the fact that the machine learning models do not have adequate historical data to predict the workload yet. In this scenario, review hours should still be allocated to these new work types regardless of their predicted workload.

We investigate the allocations obtained from both the EG solutions and the leximin-optimal solutions. In addition, we consider two cases for buyers' budget: we first assume buyers have unit budgets, and then assume that buyers' budgets are equal to their demands. Note that in the former case, due to Theorem~\ref{thm:leximin-eg-equiv}, the EG solutions and the leximin-optimal solutions coincide, and thus we omit one in the presentation in Figure~\ref{fig:allocation-arabic-ww49}.

\begin{figure}[ht]
    \centering
    \begin{subfigure}[b]{0.44\textwidth}
        \centering
        \includegraphics[width=\textwidth]{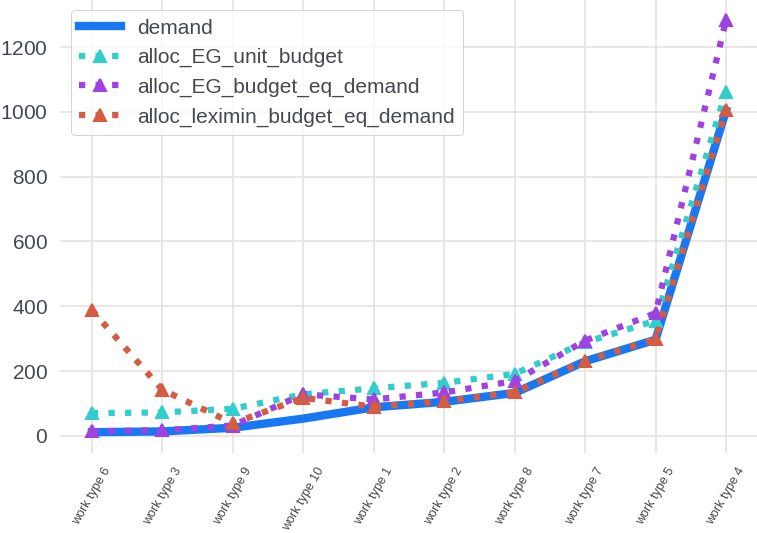}
        \caption{plot of demand and allocation} \label{fig:allocation-arabic-ww49-plot}
    \end{subfigure} %
    \hspace{.05\textwidth}
    \begin{subfigure}[b]{0.44\textwidth}
        \centering
        \includegraphics[width=\textwidth]{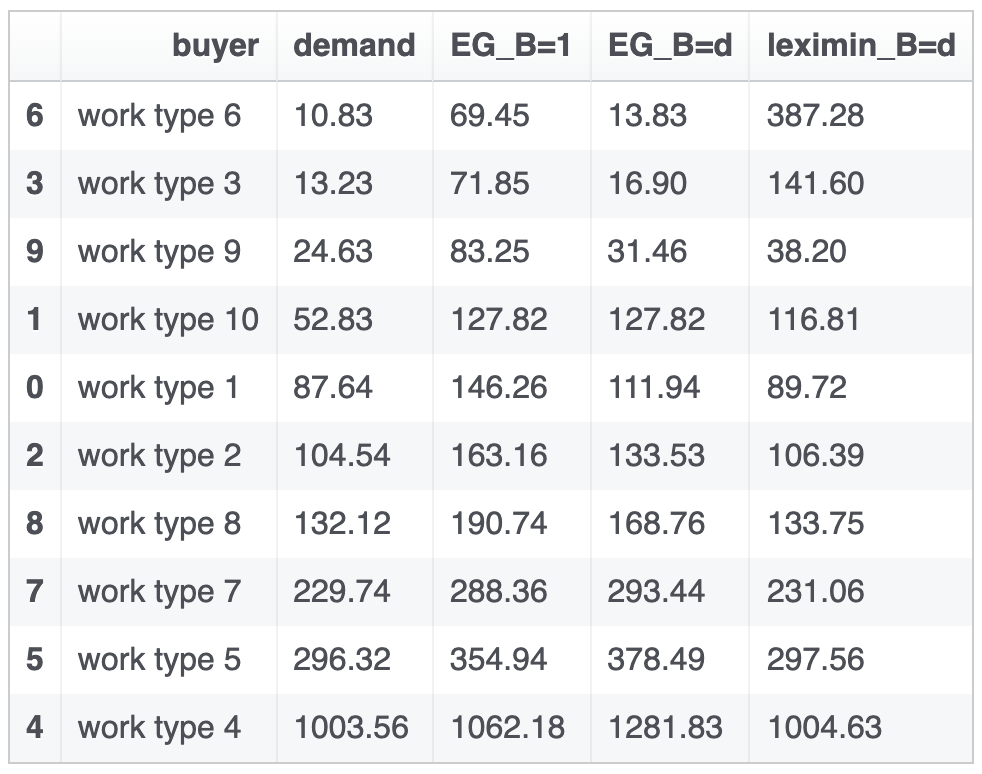}
        \caption{table of demand and allocation} \label{fig:allocation-arabic-ww49-table}
    \end{subfigure}
    \caption{Allocation under different solution concepts and budget setups. In the table columns, ``B=1'' stands for the case where all buyers have budget of one; and ``B=d'' stands for the case where budgets are equal to demands. The buyers are sorted by their demands in an increasing order.} \label{fig:allocation-arabic-ww49}
\end{figure}

We start by examining the EG solutions and we in particular observe that they satisfy properties similar to those stated in Proposition~\ref{prop:fair-add-or-multi}. When budgets are equal to ones, the oversupply from partners is equally divided among the work types. That is, the oversupply is distributed fairly in an \emph{additive} sense. More specifically, except for ``work type 10'' (which is the only one that is compatible with the ``partner 1''), all other work types are allocated an additional $60$ hours on top of their demands.

However, when buyers' budgets are set to be the same as their demand, we notice that the oversupply is allocated fairly in a more \emph{multiplicative fashion} under the EG solution. That is, work types with higher demand will receive higher amounts of allocation exceeding their demand. In our experiment, we observe that again except for the ```work type 10'', all other work types receive an allocation that is about $128\%$ of their demands. 

Lastly, the leximin-optimal solution when budgets are equal to demands does not perform well as it allocates a significant amount of oversupply to the buyer with the smallest amount of demand (or equivalently budget). The reason for this is that under leximin optimality objective, in order to maximize the minimum budget adjusted utility, the buyer with the smallest budget (or demand) will end up with the highest amount of utility (i.e., allocation). However, this solution, although is fair from an egalitarian point of view, does not align well with the business purposes here.

\section{Multiple Time Periods in a Single Program} \label{sec:multi-time}

In the robust content review problem, we would typically have constraints at several levels of the problem: a given review site may have an overall maximum capacity that they can allocate across the planning period, but then they also typically have tighter per-time-period staffing constraints.
To accommodate such constraints, we extend the basic Fisher market model to allow each item to have time-based constraints. More specifically, we will consider a model with $T$ discrete time steps, and each item $j$ will have an \emph{overall} supply $s_j$, as well as a {\em per-time period} supply $s_j^t$. 

The allocation vector now has a time component and we denote by $x^t_{ij}$ the amount of item $j$ allocated to buyer $i$ at time $t$. Similarly, the demand vector now also a time component, where $d_i^t$ is the demand of buyer $i$ at time $t$, and we denote $d_i\coloneqq \sum_t d_i^t$ as the total demand of buyer $i$.

Similar to the previous section, we may consider two solution concepts and investigate their equivalence under the special case of unit budget and binary valuation. In terms of buyers' overall utilities, we consider two possibility to combine their utilities at each time period, where we denote the utility of buyer $i$ at time $t$ as 
$$u_i^t(x_i^t) = \sum_{j} v_{ij} ( \sum_t x^t_{ij}) - d_i^t. $$

In the first approach, the utility of each buyer is simply the sum of the utilities at different times, whereas in the second approach,  the utility of each buyer is the geometric mean of the utilities at different times. In terms of fairness, the first approach is a direct extension of the previous section, aiming to be fair to all the buyers; and the second approach targets fairness at a more granular level, aiming to be fair to all the buyers across all the time periods.

\subsection{Total utility as the sum across time periods} \label{sec:multi-time-sum}

For one possibility, we assume the overall utility of buyer $i$ is
$$u_i(x_i) = \sum_{t} u_i^t(x_i^t)$$
Hence, we have the following program, which maximizes the budget-weighted utilities over the set of feasible allocations under time-based constraints. Note that this formulation also applies to the special case where each buyer only has an overall demand throughout all time periods (i.e., no time-specific demand).

\begin{equation} \label{prog:time_network}
\begin{array}{rrll|l}
  \displaystyle \max_{x\geq 0}\quad & \multicolumn{3}{l}{ \displaystyle\sum_i B_i\log (\sum_{j} v_{ij} (\sum_t x_{ij}^t) - d_i)} & \textrm{Dual variables}\\
  s.t.\quad &\displaystyle  \sum_i x^t_{ij} &\leq  s_j^t, & \forall j=1,\ldots,m, \forall t=1,\ldots,T \quad & \lambda^t_j \\
      \quad &\displaystyle  \sum_{t,i} x^t_{ij} &\leq  s_j, & \forall j=1,\ldots,m \quad & \lambda_j
\end{array}
\tag{EG-Time-sum}
\end{equation}

It follows from stationarity of the KKT conditions that the following must hold:
\begin{align}
  B_i \frac{v_{ij}}{u_i} \leq \lambda_j^t + \lambda_j, \quad \forall i,j,t
  \label{eq:hierarchichal eg stationarity}
\end{align}
Note that here, we omit the dual variables corresponding to the non-negative constraints and thus, the above must hold with equality if $x_{ij}^t > 0$ due to complementary slackness. 

\begin{lemma}
    Let $x,\lambda$ be a solution to \eqref{prog:time_network}. Consider the pair $(x,p)$, where the price $p_j^t = \lambda_j^t + \lambda_j$ on each item $j$ at a given time period $t$. Then, every buyer spends their entire demand-inflated budget. That is,
    $$\sum_{t,j}x_{ij}^t p_j^t = B_i (1+\frac{d_i}{u_i}), \quad \forall i=1,\dots,n$$
\end{lemma}

\begin{proof}
    Multiplying (\ref{eq:hierarchichal eg stationarity}) by $x_{ij}^t$ and summing over $t$ and $j$, we get that for all buyer $i$,
    $$\sum_{t,j}x_{ij}^t p_j^t  = \sum_{t,j}x_{ij}^t (\lambda_j^t + \lambda_j) = B_i \sum_{t,j}x_{ij}^t \frac{v_{ij}}{u_i} = B_i (\frac{u_i+d_i}{u_i}) = B_i (1+\frac{d_i}{u_i}).$$
\end{proof}

Thus, when buyers have no hard demand, the market clears if we set the price of item $j$ at time $t$ equal to $p_j^t = \lambda_j^t + \lambda_j$.

To show that the allocation-price pair $(x,p)$ is a market equilibrium, we need to show the  condition that an item has price $p_j^t$ strictly greater than $0$ only if $\sum_{i,t} x_{ij}^t = s_j$. This is not guaranteed in our setting. Instead, we have that $p_j^t > 0$ only if \emph{either} $\sum_i x_{ij}^t = s_j$ \emph{or} $\sum_{t,i} x_{ij}^t = s_j^t$. Thus, we get a market equilibrium if we generalize the price complementarity requirement to be that an item must have a nonzero price only if it is supply-constrained at some layer of its supply hierarchy. More formally, we have the following theorem:

\begin{theorem}
    Let $x,\lambda$ be a solution to \eqref{prog:time_network}. Consider the pair $(x,p)$, where the price $p_j^t = \lambda_j^t + \lambda_j$ on each item $j$ at a given time period $t$. This pair is a market equilibrium under the condition that an item has $p_j^t>0$ only if \emph{either} $\sum_i x_{ij}^t = s_j^t$ \emph{or} $\sum_{t,i} x_{ij}^t = s_j$, and under the following definition of demand sets:
    $$D_i(p) \coloneqq \arg \max_{x_i\ge 0} \bigg\{ u_i(x_i) \coloneqq \sum_t u_i^t(x_i^t) \, \bigg|\, \sum_{j,t} p_j^t x_{ij}^t \le B_i (1+\frac{d_i}{u_i(x_i)}) \bigg\}.$$
  \label{thm:eg-time}
\end{theorem}

\begin{proof}
    Rearranging the terms of \eqref{eq:hierarchichal eg stationarity}, we have
    $$ \frac{v_{ij}}{p_j^t} \le \frac{u_i}{B_i}, \quad \forall i,j,t,$$
    where equality holds when $x_{ij}^t>0$ due to the complementary slackness condition on $x_{ij}^t$. Thus, we get that buyers only buy items that give maximal bang-per-buck. It follows that each buyer gets a bundle in their demand set.
    
    In addition, by complementary slackness, we have that $\lambda_j^t=0$ unless $\sum_i x^t_{ij} = s_j^t$, and similarly $\lambda_j=0$ unless $\sum_{t,i} x^t_{ij} = s_j$. It follows that $p_j^t = \lambda_j +\lambda_j^t>0$ only if \emph{either} $\sum_i x_{ij}^t = s_j^t$ \emph{or} $\sum_{t,i} x_{ij}^t = s_j$.
\end{proof}

A very related theorem is shown by \citet{jain2010eisenberg}: the EG program associated to any market which can be formulated as a network flow problem with a single source, and a sink for each buyer, where the utility of a buyer is the sum of mass on paths to the corresponding sink, and with the price on a path being equal to the sum of prices on the edges in the path, can be solved by maximizing the sum of the logs of these buyer utilities, and this yields a market equilibrium. \eqref{prog:time_network} fits in this framework, except that we have heterogeneous edge weights that go into the utility function.

Algorithm~\ref{alg:leximin} can be easily adjusted to obtain the leximin-optimal solution in this setting with time-based constraints. In the following, we investigate whether results in Section~\ref{sec:demand-compare} extend to the multiple-time setting. Since the multiple-time setting is a generalization, the negative results in Section~\ref{sec:demand-compare} immediately extend. Thus, we focus on the positive result and in particular we show that when all buyers have budgets of ones and when valuations are binary, the set of MNW solutions and the set of leximin-optimal solutions coincides.

\begin{theorem} \label{thm:leximin-eg-equiv-time-sum}
    With binary valuations and unit budget, for \eqref{prog:time_network}, the set of leximin-optimal solutions is the same as the set of MNW solutions under Assumption~\ref{ass:pos}.
\end{theorem}

\begin{proof}
    Let $x$ be an allocation that is an EG solution (or equivalently an MNW solution since all buyers have budgets of one), we want to show that it is also leximin-optimal. Assume by contradiction that it is not. Let $\tilde x$ be a leximin-optimal solution. Let $r\coloneqq (u_i(x_i))_{i\in [n]}$ and $\tilde r\coloneqq (u_i(\tilde x_i))_{i\in [n]}$ denote the vectors of budget-weighted utilities corresponding allocation $x$ and $\tilde x$, receptively. First observe that by optimality we must have the following.
    
    \begin{claim} \label{cl:util-sum-same-time-sum}
        $\sum_{i} r_i(x_i) = \sum_i \tilde r_i(\tilde x_i) = \sum_j \min(s_j, \sum_t s_j^t) - \sum_i d_i$.
    \end{claim}
    
    Then, by Assumption \ref{ass:pos}, we can apply the same proof idea as given in \cite{halpern2020fair}. Let $o_1, o_2, \cdots, o_n$ be an enumeration of $[n]$ such that $r_{o_1}\le r_{o_2}\le \cdots \le r_{o_n}$, and similarly let $\tilde o_1, \tilde o_2, \cdots, \tilde o_n$ be an enumeration such that $\tilde r_{\tilde o_1}\le \tilde r_{\tilde o_2}\le \cdots \le \tilde r_{\tilde o_n}$. Since $\tilde r\succ_{\textup{lex}} r$ by assumption, we can let $k\in [n]$ be the smallest index so that $\tilde r_{\tilde o_k} > r_{o_k}$. Then by Claim \ref{cl:util-sum-same-time-sum}, we know there must be another index $i>k$ so that $\tilde r_{\tilde o_i} < r_{o_i}$ and we let $k'$ be the smallest such index. Let $N\coloneqq \{o_1, o_2, \cdots, o_{k'-1}\}$. Then, by definition of the enumeration $o$ and our choice of $k'$, we have $\sum_{i\in N} r_i < \sum_{i\in N} \tilde r_i$. Thus, there exists buyers $i\in N$, $i'\in [n]\setminus N$, and item $j\in [m]$ such that $v_{ij} = v_{i'j} = 1$ and $x_{i'j}^t>0$ for some $t\in [T]$. Also, note that we must have $r_i< r_{i'}$ by definition of $N$. Hence, by Assumption \ref{ass:pos}, one can re-allocate some item $j$ from buyer $i'$ to $i$ at time period $t$ without affecting feasibility. However, this re-allocation will increase the Nash Social Welfare, which contradicts the assumption that $x$ is an MNW solution. 
\end{proof}

\subsection{Total utility as the geometric mean across time periods} \label{sec:util-geo-mean}

For the second possibility, we assume the overall utility of buyer $i$ is
$$u_i(x_i) = \big( \prod_t u_i^t(x_i^t) \big)^{\frac{1}{T}}.$$
In this case, the logarithm of the budget-weighted geometric mean of all buyers' utility becomes
$$\log\left( \big(\prod_i [u_i(x_i)]^{B_i} \big)^{\frac{1}{n}}\right) = \frac{1}{T} \frac{1}{n} \sum_{i} B_i \sum_{t} \log(u_i^t(x_i^t))$$

To maximize the budget-weighted utilities over the set of feasible allocations under time-based constraints, we have the following program. 
\begin{equation} \label{prog:time_network_gm}
\begin{array}{rrll|l}
  \displaystyle \max_{x\geq 0}\quad & \multicolumn{3}{l}{ \displaystyle \sum_i B_i \cdot \frac{1}{T}\sum_t \log (\sum_{j} v_{ij}  x_{ij}^t - d_i^t)} & \textrm{Dual variables}\\
  s.t.\quad &\displaystyle  \sum_i x^t_{ij} &\leq  s_j^t, & \forall j\in [m], \forall t\in [T] \quad & \lambda^t_j \\
      \quad &\displaystyle  \sum_{t,i} x^t_{ij} &\leq  s_j, & \forall j\in[m] \quad & \lambda_j
\end{array}
\tag{EG-Time-geo-mean}
\end{equation}

In the setting, stationarity of the KKT conditions gives us the following relationship:
\begin{align}
    B_i \cdot \frac{1}{T} \cdot \frac{v_{ij}}{u_i^t} \leq \lambda_j^t + \lambda_j, \quad \forall i,j,t
  \label{eq:hierarchichal eg stationarity gm}
\end{align}
where equality holds whenever $x_{ij}^t > 0$ due to complementary slackness. 

\begin{lemma} \label{lem:budget-inflated-demand-time}
    Let $x,\lambda$ be a solution to \eqref{prog:time_network_gm}. Consider the pair $(x,p)$, where the price $p_j^t = \lambda_j^t + \lambda_j$ on each item $j$ at a given time period $t$. Then, every buyer spends their entire demand-inflated budget. That is,
    $$\sum_{t,j}x_{ij}^t p_j^t = B_i \left[ \frac{1}{T} \sum_t (1+\frac{d_i^t}{u_i^t}) \right], \quad \forall i=1,\dots,n$$
\end{lemma}

\begin{proof}
    Multiplying \eqref{eq:hierarchichal eg stationarity gm} by $x_{ij}^t$ and summing over $t$ and $j$, we get that for all buyer $i$,
    $$\sum_{t,j}x_{ij}^t p_j^t = B_i \cdot \frac{1}{T} \sum_{t,j}x_{ij}^t \frac{v_{ij}}{u_i^t} = B_i \cdot \frac{1}{T} \sum_t (\frac{u_i^t+d_i^t}{u_i^t}) = B_i \left[ \frac{1}{T} \sum_t (1+\frac{d_i^t}{u_i^t}) \right].$$
\end{proof}

Thus, when buyers have no hard demand at all time periods (i.e., $d_i^t=0$ for all $i, t$), the market clears if we set the price of item $j$ at time $t$ equal to $p_j^t = \lambda_j^t + \lambda_j$. 

To show that the allocation-price pair $(x,p)$ is a market equilibrium, not only do we need to modify the price complementarity requirement as in the previous subsection, we also need to modify the definition of demand sets due to Lemma \ref{lem:budget-inflated-demand-time}. 

\begin{theorem}
    Let $x,\lambda$ be a solution to \eqref{prog:time_network_gm}. Consider the pair $(x,p)$, where the price $p_j^t = \lambda_j^t + \lambda_j$ on each item $j$ at a given time period $t$. This pair is a market equilibrium under the condition that an item has $p_j^t>0$ only if \emph{either} $\sum_i x_{ij}^t = s_j^t$ \emph{or} $\sum_{t,i} x_{ij}^t = s_j$, and under the following definition of demand sets:
    $$D_i^t(p) = \arg \max_{x_i\ge 0} \bigg\{u_i(x_i) \coloneqq \bigg( \prod_t (u_i^t(x_i^t) ) \bigg)^{\frac{1}{T}} \bigg| \sum_{j,t} p_j^t x_{ij}^t \leq B_i \bigg[ \frac{1}{T} \sum_t (1+\frac{d_i^t}{u_i^t(x_i^t)}) \bigg] \bigg\}.$$
\end{theorem}

\begin{proof}
    Note that the utility function is no longer linear but is concave since the geometric mean function is concave. To show that every buyer is allocated a bundle that is in their demand set, we instead want to show that $x_{ij}^t > 0$ only if it has the maximal \emph{marginal} bang-per-buck, where the marginal utility of item $j$ for buyer $i$ at time $t$ is
    $$\frac{\partial u_i(x_i)}{\partial x_{ij}^t} = \bigg(\prod_{t': t'\neq t} u_i^{t'} (x_i^{t'}) \bigg) v_{ij}^t = \frac{u_i(x_i)}{u_i^t(x_i^t)} v_{ij}^t.$$
    Rearranging the terms of \eqref{eq:hierarchichal eg stationarity gm}, we have
    $$\frac{v_{ij}}{p_j^t} \leq \frac{u_i^t}{B_i} \Rightarrow \frac{u_i}{u_i^t}v_{ij} / p_j^t \leq \frac{u_i}{B_i}, \quad \forall i,j,t,$$
    where equality holds when $x_{ij}^t>0$ due to the complementary slackness condition on $x_{ij}^t$. Thus, we get that buyers only buy items that give maximal marginal bang-per-buck, as desired. The proof for the second condition is the same as that in the proof of Theorem~\ref{thm:eg-time}.
\end{proof}

In the following, we define what makes a solution leximin-optimal in this setting. Instead of having a vector of buyers' budget-weighted utilities, we now have a vector of buyers' budget-weighted utilities at all time periods. More specifically, for any allocation $x$, the vector $r(x)$ now has a time component: $r_i^t(x) = u_i^t(x_i^t)^{B_i}$; and an allocation $x$ is \emph{leximin-optimal} if there is no other allocation $\tilde x$ such that $r(\tilde x) \succ_{\textup{lex}} r(x)$. With this definition, when all buyers have unit budgets, \eqref{prog:time_network_gm} is the same as \eqref{prob:main} if we treat each buyer at each time period as an individual buyer. Thus, Theorem~\ref{thm:leximin-eg-equiv} extends naturally.

\begin{corollary}
    With binary valuations and unit budget, for \eqref{prog:time_network_gm}, the set of leximin-optimal solution is the same as the set of MNW solutions under Assumption~\ref{ass:pos}.
\end{corollary}

\subsection{Application}\label{sec:app-multi-time}

For application, we continue our investigation on the same market. Instead of a single time period, we now focus on four consecutive time periods. We observe from our data that the supply from our partners is constant across the four time periods. This is expected given that partners usually have a fixed amount of facilities to accommodate their employees. On the other hand, we observe that for some work types, their forecasted demands have either an upward or a downward trend from one time period to the next (see Figure~\ref{fig:demand-trend-arabic}).

\begin{figure}[ht]
    \centering
    \includegraphics[width=.7\textwidth]{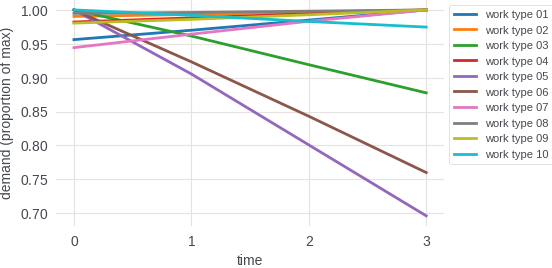}
    \caption{The trend of demand (i.e., forecasted workload) of different work types in the market from four consecutive time periods. Since the demands of different work types have different magnitude, we normalized them by displaying the proportion of the demand with respect to the maximum demand for each work type.}
    \label{fig:demand-trend-arabic}
\end{figure}

Due to the fact that variations are on the demand side, we use the second utility definition (as in Section~\ref{sec:util-geo-mean}) for our experiments. Moreover, for the following, we only focus on the case where buyers have a budget of one. The reason for this simplification is that the qualitative results are similar to those under the case where budgets are equal to demands.

The allocation results are shown in Figure~\ref{fig:allocation-trend-smooth-none}. Note that we omit ``partner 1'', since it can only work on the ``work type 10'', and thus intuitively all of its reviewing capacity should be allocated toward that work type.

As one can see from Figure~\ref{fig:allocation-trend-smooth-none}, for the most part, the allocation for each buyer-seller pair is rather smooth from time to time. However, for the ``work type 05'', its allocation from ``partner 3'' partner goes down, but it naturally follows from the fact that its demand has a downward trend. However, for ``work type 04'', although its demand is relatively constant, its allocation from both partners vary, with one having an upward trend and the other having a downward trend.

\begin{figure}[ht]
    \centering
    \includegraphics[width=.98\textwidth]{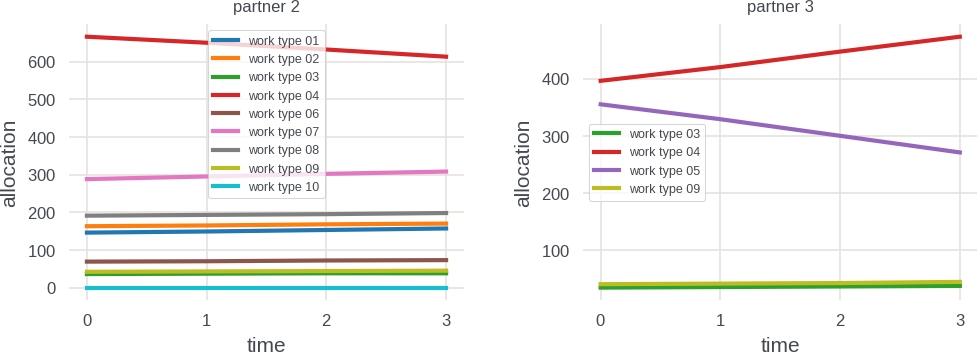}
    \caption{For each partner in the market, the plot shows the allocation to its compatible work types from all time periods. We omit ``partner 1'' since its allocation is immediate as it is only compatible to one work type.}
    \label{fig:allocation-trend-smooth-none}
\end{figure}

Because of the variability in allocation from one to the next, the solution here has the drawback that partners have to shift their staffing requirements from time to time, which could be costly from an operational point of view. As a result, it would be beneficial to the partners if we can obtain a solution that minimizes staffing change as much as possible. We discuss more in detail on markets' desire for smooth allocation in the next section.

\section{Smooth Allocation Across Multiple Discrete Time Periods} \label{sec:multi-time-smooth}

In the previous section, we saw how to allow each item to have per-time-period and overall supply. However, the resulting allocations may still exhibit an effect which is often considered undesirable in practice: the allocation may have large swings across time, even when the utility from this is low. Consider the following robust content review example:
\begin{example}
There are two review categories (buyers) and two review sites (items), with three time steps. Each review site has per-period supply of one, and no cross-time supply constraint. All valuations are one. There are many optimal allocations, but we will consider three in particular: 
\begin{itemize}
    \item a {\em high variation} allocation $x_{11} = (1,0,1), x_{12} = (0,1,0), x_{21} = (0,1,0), x_{22}= (1,0,1)$; 
    \item a {\em low variation} allocation $x_{11}' = (1,1,1), x_{22}' = (1,1,1)$ with all other $x'$ equal to zero; 
    \item another low-variation allocation where $\hat x_{ij} = (0.5, 0.5, 0.5)$ for all $i,j$.
\end{itemize}

In the high-variation allocation, each review site gets a non-smooth allocation: in time step $2$ and time step $3$, they flip their entire review category from one to the other. In practice, this is inefficient, due to staff context switching or even requiring different staff allocations. This can lead to less efficient reviewing for a given content category. In contrast, the two low-variation allocations give each review site a nice constant amount of work from each review category. We do not necessarily prefer $\hat x$ or $x'$ to each other; both are equally ``smooth''. 
\end{example}

From the example above, we can clearly see the problem with the market equilibrium solutions, which is that they do not take into account market's preferences for low-variation solutions.

To model the preference for low variation, one approach is to add constraints forcing the variation to be within a certain range. For instance, if we want the allocation from one time period to the next to be within $5\%$, we can add constraints
\begin{align*}
    x_{ij}^{t+1} \ge 0.95\cdot x_{ij}^t, &\quad \forall i\in [n], j\in [m], t\in [T-1] \\
    x_{ij}^{t+1} \le 1.05\cdot x_{ij}^t, &\quad \forall i\in [n], j\in [m], t\in [T-1].
\end{align*}
However, it is not hard to see that one issue with the constraints approach is that when demands are high (i.e., small amount of oversupply to distribute) and when variation in demand from time to time is also high, the problem could become infeasible.

Alternatively, we can include (or subtract, to be more precise) a penalty term in the objective. We assume the penalty term takes the following form:
$$\gamma \sum_{i,j,t} R(x_{ij}^{t+1}, x_{ij}^t),$$
where $R_{jt}$ is any convex function that measures the discrepancy between two quantities and $\gamma$ is a hyperparameter that captures the strength in terms of the market's desire for smooth allocation. Since $R_{jt}$ is the penalty function for non-smooth allocation, we further require that $R_{jt}(x_{ij}^{t+1},x_{ij}^{t}) = 0$ whenever $x_{ij}^{t+1}=x_{ij}^{t}$. Hence, we have the following program.

\begin{equation}\label{prog:smooth_time_network}
\begin{array}{rrll|l}
    \displaystyle \max_{x\geq 0}\quad & \multicolumn{3}{l}{ \displaystyle \sum_{i} B_i \sum_t \log (\sum_{j} x_{ij}^t v_{ij} - d_i^t) - \gamma \sum_{i,j,t} R(x_{ij}^{t+1}, x_{ij}^t) } \\
    \text{s.t.}\quad &\displaystyle  \sum_i x^t_{ij} &\leq  s_j^t, & \forall j=1,\ldots,m, \forall t=1,\ldots,T,\quad \\
    \quad &\displaystyle  \sum_{t,i} x^t_{ij} &\leq  s_j, & \forall j=1,\ldots,m,\quad 
\end{array}
\tag{EG-Smooth}
\end{equation}

Below, we list a couple of interesting choices for the penalty function.
\begin{itemize}
\item {\em Absolute deviation}: $R_{jt}(x_{i,j}^{t+1}, x_{i,j}^{t}) = | x_{i,j}^{t+1}- x_{i,j}^{t} |$. 
\item {\em Kullback-Leibler}: $R_{jt}(x_{i,j}^{t+1}, x_{i,j}^{t}) = \max(x_{i,j}^{t+1} \log \frac{x_{i,j}^{t+1}}{x_{i,j}^{t}}, x_{i,j}^{t} \log \frac{x_{i,j}^{t}}{x_{i,j}^{t+1}})$. 
\end{itemize}

\subsection{Application} \label{sec:app-multi-time-smooth}

To continue our experiment in Section~\ref{sec:app-multi-time}, we now include an additional penalty term in the objective. We show the results in Figure~\ref{fig:allocation-trend-smooth-abs-dev} using the absolute difference penalty function with appropriate smoothness parameters. We would like to point out the qualitative results using the KL divergence penalty function is similar and we omit them here. 

\begin{figure}[ht]
    \centering
    \includegraphics[width=.98\textwidth]{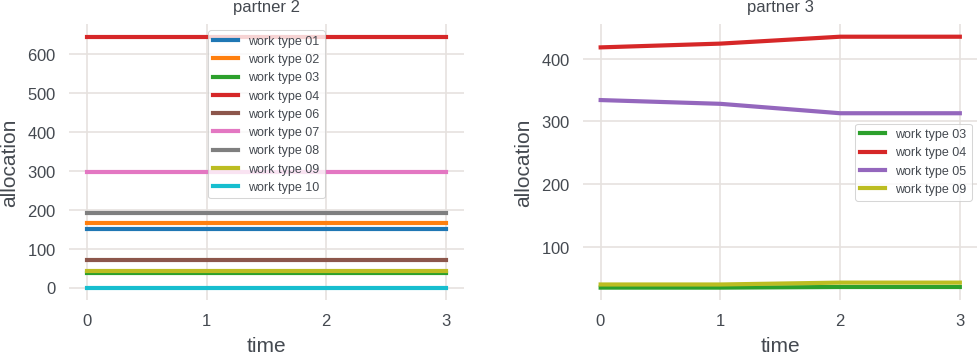}
    \caption{This is a similar plot to Figure~\ref{fig:allocation-trend-smooth-none}, but this figure plots the allocation with absolute difference penalty function in the objective when the smoothness parameter $\gamma$ is set to be $0.005$.}
    \label{fig:allocation-trend-smooth-abs-dev}
\end{figure}

Note that now for both ``work type 04'' and ``work type 05'', the allocations to ``partner 3'' are much smoother as compared those previously obtained as shown in Figure~\ref{fig:allocation-trend-smooth-none}.

As the smoothness parameter becomes larger, the dominating component in the objective of \eqref{prog:smooth_time_network} shifts from the utility term to the smoothness penalty term. Hence, if we were to set $\gamma = 0.01$, the resulting allocation becomes constant for each (work type, partner) pair. However, such a smooth allocation is at the expense of wasting partners' reviewing capacity. Thus, the smoothness parameter should be carefully chosen to balance different business objectives.

\medskip
\noindent\textbf{Acknowledgement.} The authors would like to thank Martin Bscheider and Andrew Martens for their support throughout the work.

\bibliographystyle{plainnat}
\bibliography{refs}

\begin{thebibliography}{27}
\providecommand{\natexlab}[1]{#1}
\providecommand{\url}[1]{\texttt{#1}}
\expandafter\ifx\csname urlstyle\endcsname\relax
  \providecommand{\doi}[1]{doi: #1}\else
  \providecommand{\doi}{doi: \begingroup \urlstyle{rm}\Url}\fi

\bibitem[Arrow and Intriligator(2005)]{arrow2005handbook}
Kenneth Arrow and Michael Intriligator.
\newblock Handbook of mathematical economics.
\newblock Technical report, Elsevier, 2005.

\bibitem[Aziz and Brown(2020)]{aziz2020random}
Haris Aziz and Ethan Brown.
\newblock Random assignment under bi-valued utilities: Analyzing
  hylland-zeckhauser, nash-bargaining, and other rules.
\newblock \emph{arXiv preprint arXiv:2006.15747}, 2020.

\bibitem[Aziz and Rey(2020)]{aziz2020almost}
Haris Aziz and Simon Rey.
\newblock Almost group envy-free allocation of indivisible goods and chores.
\newblock \emph{Proceedings of the 29th International Joint Conference on
  Artificial Intelligence (IJCAI)}, 2020.

\bibitem[Bassett(2022)]{cm-ai}
Caitlin Bassett.
\newblock Will {AI} take over content moderation?
\newblock \emph{Mind {M}atters {N}ews}, 2022.
\newblock URL
  \url{https://mindmatters.ai/2022/01/will-ai-take-over-content-moderation/}.

\bibitem[Caragiannis et~al.(2019)Caragiannis, Kurokawa, Moulin, Procaccia,
  Shah, and Wang]{caragiannis2019unreasonable}
Ioannis Caragiannis, David Kurokawa, Herv{\'e} Moulin, Ariel~D Procaccia,
  Nisarg Shah, and Junxing Wang.
\newblock The unreasonable fairness of maximum nash welfare.
\newblock \emph{ACM Transactions on Economics and Computation (TEAC)},
  7\penalty0 (3):\penalty0 1--32, 2019.

\bibitem[Eisenberg and Gale(1959)]{eisenberg1959consensus}
Edmund Eisenberg and David Gale.
\newblock Consensus of subjective probabilities: The pari-mutuel method.
\newblock \emph{The Annals of Mathematical Statistics}, 30\penalty0
  (1):\penalty0 165--168, 1959.

\bibitem[Facebook(2022)]{facebook-cs}
Facebook.
\newblock Facebook community standards, 2022.
\newblock URL \url{https://transparency.fb.com/policies/community-standards/}.
\newblock Online; accessed 7-February-2022.

\bibitem[Foley(1966)]{foley1966resource}
Duncan~Karl Foley.
\newblock \emph{Resource allocation and the public sector}.
\newblock Yale University, 1966.

\bibitem[Freeman et~al.(2019)Freeman, Sikdar, Vaish, and
  Xia]{freeman2019equitable}
Rupert Freeman, Sujoy Sikdar, Rohit Vaish, and Lirong Xia.
\newblock Equitable allocations of indivisible goods.
\newblock \emph{arXiv preprint arXiv:1905.10656}, 2019.

\bibitem[Garcelon et~al.(2021)Garcelon, Avadhanula, Lazaric,
  et~al.]{garcelon2021top}
Evrard Garcelon, Vashist Avadhanula, Alessandro Lazaric, et~al.
\newblock Top $ k $ ranking for multi-armed bandit with noisy evaluations.
\newblock \emph{arXiv preprint arXiv:2112.06517}, 2021.

\bibitem[Gillespie(2020)]{gillespie2020content}
Tarleton Gillespie.
\newblock Content moderation, ai, and the question of scale.
\newblock \emph{Big Data \& Society}, 7\penalty0 (2):\penalty0
  2053951720943234, 2020.

\bibitem[Haimovich et~al.(2020)Haimovich, Karamshuk, Leeper, Riabenko, and
  Vojnovic]{haimovich2020scalable}
Daniel Haimovich, Dima Karamshuk, Thomas~J Leeper, Evgeniy Riabenko, and Milan
  Vojnovic.
\newblock Scalable prediction of information cascades over arbitrary time
  horizons.
\newblock \emph{arXiv preprint arXiv:2009.02092}, 2020.

\bibitem[Halpern et~al.(2020)Halpern, Procaccia, Psomas, and
  Shah]{halpern2020fair}
Daniel Halpern, Ariel~D Procaccia, Alexandros Psomas, and Nisarg Shah.
\newblock Fair division with binary valuations: One rule to rule them all.
\newblock In \emph{International Conference on Web and Internet Economics},
  pages 370--383. Springer, 2020.

\bibitem[Jain and Vazirani(2010)]{jain2010eisenberg}
Kamal Jain and Vijay~V Vazirani.
\newblock Eisenberg--gale markets: algorithms and game-theoretic properties.
\newblock \emph{Games and Economic Behavior}, 70\penalty0 (1):\penalty0
  84--106, 2010.

\bibitem[Jhaver et~al.(2019)Jhaver, Birman, Gilbert, and
  Bruckman]{jhaver2019human}
Shagun Jhaver, Iris Birman, Eric Gilbert, and Amy Bruckman.
\newblock Human-machine collaboration for content regulation: The case of
  reddit automoderator.
\newblock \emph{ACM Transactions on Computer-Human Interaction (TOCHI)},
  26\penalty0 (5):\penalty0 1--35, 2019.

\bibitem[Kurokawa et~al.(2018)Kurokawa, Procaccia, and
  Shah]{kurokawa2018leximin}
David Kurokawa, Ariel~D Procaccia, and Nisarg Shah.
\newblock Leximin allocations in the real world.
\newblock \emph{ACM Transactions on Economics and Computation (TEAC)},
  6\penalty0 (3-4):\penalty0 1--24, 2018.

\bibitem[Makhijani et~al.(2021)Makhijani, Shah, Avadhanula, Gocmen,
  Stier-Moses, and Mestre]{makhijani2021quest}
Rahul Makhijani, Parikshit Shah, Vashist Avadhanula, Caner Gocmen,
  Nicol{\'a}s~E Stier-Moses, and Juli{\'a}n Mestre.
\newblock Quest: Queue simulation for content moderation at scale.
\newblock \emph{arXiv preprint arXiv:2103.16816}, 2021.

\bibitem[Meta(2021)]{meta-csep}
Meta.
\newblock Community standards enforcement report, 2021.
\newblock URL
  \url{https://transparency.fb.com/data/community-standards-enforcement/}.
\newblock Online; accessed 7-February-2022.

\bibitem[Nash(1953)]{nash1953two}
John Nash.
\newblock Two-person cooperative games.
\newblock \emph{Econometrica}, 21\penalty0 (1):\penalty0 128--140, 1953.

\bibitem[Nguyen et~al.(2020)Nguyen, Shi, Ramakrishnan, Weinsberg, Lin, Metz,
  Chandra, Jing, and Kalimeris]{nguyen2020clara}
Viet-An Nguyen, Peibei Shi, Jagdish Ramakrishnan, Udi Weinsberg, Henry~C Lin,
  Steve Metz, Neil Chandra, Jane Jing, and Dimitris Kalimeris.
\newblock Clara: confidence of labels and raters.
\newblock In \emph{Proceedings of the 26th ACM SIGKDD International Conference
  on Knowledge Discovery \& Data Mining}, pages 2542--2552, 2020.

\bibitem[Osofsky(2019)]{osofsky2019}
Justin Osofsky.
\newblock Our commitment to our content reviewers.
\newblock 2019.
\newblock URL
  \url{https://about.fb.com/news/2019/02/commitment-to-content-reviewers/}.

\bibitem[Panageas et~al.(2021)Panageas, Tr{\"o}bst, and
  Vazirani]{panageas2021combinatorial}
Ioannis Panageas, Thorben Tr{\"o}bst, and Vijay~V Vazirani.
\newblock Combinatorial algorithms for matching markets via nash bargaining:
  One-sided, two-sided and non-bipartite.
\newblock \emph{arXiv preprint arXiv:2106.02024}, 2021.

\bibitem[Plaut and Roughgarden(2020)]{plaut2020almost}
Benjamin Plaut and Tim Roughgarden.
\newblock Almost envy-freeness with general valuations.
\newblock \emph{SIAM Journal on Discrete Mathematics}, 34\penalty0
  (2):\penalty0 1039--1068, 2020.

\bibitem[Rosen(2021)]{rosen2021}
Guy Rosen.
\newblock Community standards enforcement report, third quarter 2021.
\newblock 2021.
\newblock URL
  \url{https://about.fb.com/news/2021/11/community-standards-enforcement-report-q3-2021/}.

\bibitem[Schoolov(2021)]{cnbc-cm}
Katie Schoolov.
\newblock Why content moderation costs billions and is so tricky for facebook,
  twitter, youtube and others.
\newblock \emph{CNBC}, 2021.
\newblock URL
  \url{https://www.cnbc.com/2021/02/27/content-moderation-on-social-media.html}.

\bibitem[Varian et~al.(1974)]{varian1974equity}
Hal~R Varian et~al.
\newblock Equity, envy, and efficiency.
\newblock \emph{Journal of Economic Theory}, 9\penalty0 (1):\penalty0 63--91,
  1974.

\bibitem[Wilson(1998)]{wilson1998fair}
Stephen~J Wilson.
\newblock Fair division using linear programming.
\newblock \emph{preprint, Departement of Mathematics, Iowa State University},
  1998.

\end{thebibliography}


\end{document}